\documentclass[aip,amsmath,amssymb,preprint]{revtex4-1}

\usepackage{graphicx}
\usepackage{dcolumn}
\usepackage{bm}

\usepackage[utf8]{inputenc}
\usepackage[T1]{fontenc}
\usepackage{mathptmx}

\usepackage{url}
\usepackage{setspace}

\usepackage{color, colortbl}

\definecolor{c1}{RGB}{92,89,191}
\definecolor{c2}{RGB}{191,92,98}
\definecolor{c3}{RGB}{98,191,92}

\newtheorem{lemma}{Lemma}
\newtheorem{remark}{Remark}

\newtheorem{definition}{Definition}

\newtheorem{proof}{Proof}

\definecolor{sorren}{RGB}{0,0,0}

\begin{document}

\preprint{AIP/123-QED}

\title{Symmetry Induced Group Consensus}

\author{Isaac Klickstein}
  \email{iklick@unm.edu}
  \affiliation{Department of Mechanical Engineering, University of New Mexico, Albuquerque, NM 87106, USA} 
  \author{Louis Pecora}
  \email{louis.pecora@nrl.navy.mil}
  \affiliation{U.S. Naval Research Laboratory, Washington, District of Columbia 20375, USA}
\author{Francesco Sorrentino}
  \email{fsorrent@unm.edu}
  \affiliation{Department of Mechanical Engineering, University of New Mexico, Albuquerque, NM 87106, USA} 
 
\date{\today}
\begin{abstract}
  There has been substantial work studying consensus problems for which there is a single common final state, although there are many real-world complex networks for which the complete consensus may be undesirable.
  More recently, the concept of group consensus whereby subsets of nodes are chosen to reach a common final state distinct from others has been developed, but the methods tend to be independent of the underlying network topology.
  Here, an alternative type of group consensus is achieved for which nodes that are \emph{symmetric} achieve a common final state.
  The dynamic behavior may be distinct between nodes that are not symmetric.
  We show how group consensus for heterogeneous linear agents can be achieved via a simple coupling protocol that exploits the topology of the network. We see that  group consensus is possible on both stable and unstable trajectories.
  We observe and characterize the phenomenon of \textit{isolated group consensus}, where one or more clusters may achieve group consensus while the other clusters do not.
\end{abstract}
\maketitle
\begin{quotation}
  Consensus problems are an important topic when designing control protocols for distributed systems where it is desirable for uniform behavior between nodes like power generators in the grid, mobile robots, or autonomous vehicles.
  Sometimes though, complete consensus is not the desired behavior, but rather group consensus whereby some nodes' behaviors will coincide while others do not.
  This problem has been tackled previously using linear matrix inequalities (LMIs) or Lyapunov functions but the result can only guarantee complete consensus.
  We instead present a method derived using the automorphism group of the underlying graph which provides more granular information that splits the dynamics of consensus motion from different types of orthogonal, cluster breaking motion.
\end{quotation}
\section{Introduction}
For the past twenty years, the field of multi-agent network dynamics, and the coordination between said agents, has been investigated by researchers from a vast range of disciplines \cite{ren2005consensus, ren2008consensus, qin2013coordination}.
The applications for a well designed method of coordination are diverse, from vehicle attitudes \cite{ren2008distributed}, opinion dynamics \cite{proskurnikov2016opinion}, sensor networks \cite{he2013sats}, and communication networks \cite{lin2014constrained}.
A consensus protocol is defined as information sharing between agents.
The information shared often takes the form of relative states \cite{ren2008distributed}.
More recently, research has focused on when a network is able to achieve group consensus, i.e., some members of the network reach consensus with each other, but not necessarily with all members of the network.
Group consensus is investigated for undirected and directed networks with and without switching topology in \cite{yu2009group,yu2010group,yu2012group}.
The intra-group coupling is used as the tool to determine whether or not group consensus is achieved in \cite{qin2013group,qin2016group,qin2015Hinfinity}.
Related work on group and cluster synchronization of networks has been carried out in \cite{sorrentino2007network, williams2013experimental, pecora2014cluster, sorrentino2016complete}.
Much of the current research into group consensus \cite{hu2016reverse, hu2016couple, gao2015second, feng2014group} assumes a balanced adjacency graph between groups, that is, the sum of inter-cluster connections sum to zero.
Also, the groups, or clusters, of nodes that reach consensus with each other are defined on the network with a typical requirement being that each group satisfies some structural property, notably without exploiting any inherent graphical properties of the network.
Alternatively, in this paper, the final group consensus is achieved as an emergent property of the network topology.\\
\indent The groups of nodes we focus on are the so called \emph{orbits of the automorphism group} \cite{lauri2016topics}, that is, those nodes which are symmetric within the graph, that is, there are permutations of the nodes that leave the network unchanged.
The symmetries of a graph that describes the underlying topology of a network can be determined from its automorphism group.
This problem has been approached using contraction theory \cite{russo2011symmetries, russo2013convergence} but the method presented therein does not provide insight into the effect of various control gains that may be tuned.
Here we exploit the properties of the automorphism group of the network to define group consensus.
Following our recent work on generating graphs with desired symmetries \cite{klickstein2018generatingb,klickstein2018generating}, generating graphs with desirable symmetry properties can be done easily.
We approach this problem using the block diagonalizing technique presented in \cite{pecora2014cluster} which decouples consensus creating and consensus breaking dynamics.\\
\textcolor{sorren}{\indent Our results are in agreement with previous work which described phase synchronization in networks of coupled nonlinear Kuramoto oscillators \cite{nicosia2013remote}. They  differ from previous work on multi-consensus \cite{monaco2019multi} which considers directed graphs and Laplacian coupling.}\\

\indent The rest of this paper is organized as follows.
Some background material with respect to graph and matrix theory as well as the automorphism group of a graph is presented in Section II.
In Section III we use the block diagonalizing transformation \cite{pecora2014cluster} to determine when group consensus will occur regardless of the stability of the overall system.
We present an extensive example as well to highlight the features of our proposed method, such as the ability to yield isolated group consensus.
We conclude the paper with a summary and some future directions in Section V. 
\section{Preliminaries}
Throughout this paper we use the following standard notation; let $I_N$ denote the identity matrix of dimension $N$, let $O_{N\times M}$ denote the matrix of all zeros of dimension $N \times M$ (for brevity we will denote $O_{N \times N} = O_N$), let $\boldsymbol{1}_N$ denote the vector of all ones of length $N$ and let $\boldsymbol{0}_N$ denote the vector of all zeroes of length $N$.\\
\indent We define a graph $\mathcal{G} = (\mathcal{V}(\mathcal{G}),\mathcal{E}(\mathcal{G}))$ as consisting of two sets: a set of nodes $\mathcal{V} = \mathcal{V}(\mathcal{G}) = \left\{i | i=1,\ldots,N\right\}$ so that $|\mathcal{V}| = N$ and a set of edges $\mathcal{E} = \mathcal{E}(\mathcal{G}) \subseteq \mathcal{V}\times\mathcal{V}$ where $(i,j) \in \mathcal{E}$ if node $j$ receives a signal from node $i$ and $(i,j) \notin \mathcal{E}$ otherwise.
We say the graph is undirected if $(i,j) \in \mathcal{E}$ implies that $(j,i) \in \mathcal{E}$ and the graph is directed otherwise.
The adjacency matrix of the graph $\mathcal{G}$ is a binary matrix $A = \{A_{ij}\} \in \mathbb{R}^{N \times N}$ such that element $A_{ij} = 1$ if $(j,i) \in \mathcal{E}$ and $A_{ij} = 0$ if $(j,i) \notin \mathcal{E}(\mathcal{G})$.
Clearly, if $\mathcal{G}$ is undirected then $A$ is symmetric and if $\mathcal{G}$ is directed then $A$ is may be non-symmetric.
The degree of a node, denoted $d_i$, is the number of neighbors that the node $i$ has. 
\indent We now place a definition on what we mean by an orbit \cite{lauri2016topics}.
Define a permutation of the graph, $\pi(\mathcal{G}) = \mathcal{G}'$, where $\mathcal{V}(\mathcal{G}) = \mathcal{V}(\mathcal{G'})$, i.e., a permutation does not remove or introduce nodes, and if $(i,j) \in \mathcal{E}(\mathcal{G})$ then $(\pi(i),\pi(j)) \in \mathcal{E}(\mathcal{G}')$.
The graphs $\mathcal{G}$ and $\mathcal{G}'$ are said to be isomorphic if there exists a permutation with the above properties.
If $\mathcal{G} = \mathcal{G}'$ then the permutation $\pi$ is said to be an automorphism.
If $\pi$ is an automorphism and if $(i,j) \in \mathcal{E}(\mathcal{G})$  then $(\pi(i),\pi(j)) \in \mathcal{E}(\mathcal{G})$.
The same identity holds if $(i,j) \notin \mathcal{E}(\mathcal{G})$.
The set of automorphisms forms a permutation group acting on the nodes of a graph $\mathcal{G}$ which we denote as $(\text{Aut}(\mathcal{G}),\mathcal{V}(\mathcal{G}))$, or simply $\text{Aut}(\mathcal{G})$.
The set of all permutations in the autormorphism group will only permute certain subsets of nodes among each other.
These subsets of nodes are defined as the orbits (or equivalently `clusters') of the automorphism group.
There exists a permutation in the automorphism group that will permute any node in an orbit with any other node in the same orbit.
There also may exist trivial orbits, i.e., an orbit $k$ with a population of one.
\begin{lemma}\label{lem:com}
  Define the matrix $P$ as an $N \times N$ permutation matrix associated with a permutation $\pi \in \text{Aut}(\mathcal{G})$ (from now on, for simplicity, we will say $P \in \text{Aut}(\mathcal{G})$).
  If $A$ is the adjacency matrix of graph $\mathcal{G}$, then the permutation matrix $P$ commutes with $A$, i.e., $AP = PA$.
\end{lemma}
\begin{proof}
  Consider a permutation such that $\pi(v_i) = v_k$ and $\pi(v_j) = v_{\ell}$.
  Using the fact that $P$, the matrix representation of $\pi$, is binary and orthonormal, then $(PA)_{i\ell} = A_{k\ell}$ and $(AP)_{i\ell} = A_{ij}$.
  By definition of a symmetry, if $(v_i,v_j) \in \mathcal{E}$ then $(\pi(v_i),\pi(v_j)) = (v_k,v_{\ell}) \in \mathcal{E}$ as well and so $A_{ij} = A_{k\ell}$ which proves $AP = PA$.
\end{proof}
\indent Assuming there are $q$ orbits, we partition the nodes according to their orbits so that $\mathcal{V}_k$, $k=1,\ldots,q$ consists of the nodes in orbit $k$ and $\bigcup_{k=1}^q \mathcal{V}_k = \mathcal{V}$.
Note that for the standard definition of a partition, $\mathcal{V}_k \neq \emptyset$ and $\mathcal{V}_k \cap \mathcal{V}_l = \emptyset$, $k \neq l$.\\
\begin{remark}
  All nodes in orbit $k$ will have the same number of in-coming edges from each other orbit $l = 1,\ldots,q$, i.e., $\sum_j A_{ij}=Q_{kl}$ for each $i \in {\mathcal V}_k$ and $j \in {\mathcal V}_l$.
\end{remark}
\indent Let $\bar{i}$ denote the orbit in which node $i$ resides.
If $\bar{i} = \bar{j}$ then nodes $i$ and $j$ are in the same orbit.
Also, let $N_k = |\mathcal{V}_k|$ denote the number of nodes in orbit $k$, $\sum_{k=1}^q N_k = N$.
We assume that the nodes are numbered corresponding to the orbits so that all nodes in orbit $1$ are labelled $1,\ldots,N_1$, the nodes in orbit $2$ are labelled $N_1+1,\ldots, N_1+N_2$ and so forth.
The adjacency matrix is block-partitioned so that the $N_{\ell} \times N_k$ block $A_{\ell k}$ denotes the inter-orbit coupling from orbit $k$ to orbit $\ell$ and the $N_k \times N_k$ block $A_{kk}$ denotes the intra-orbit coupling within orbit $k$.
Here we assume both intra-orbit couplings and inter-orbit couplings to be undirected.
Unlike most work on group consensus that requires some structural property for the intra-cluster connectivity \cite{yu2012group,qin2016leaderless}, in this paper there may or may not by intra-cluster coupling.
\section{Group consensus in linear networks}

This paper is concerned with networks of agents characterized by linear dynamics.
We consider a very general case for which the agents obey the following set of equations,

\begin{equation} \label{eq:gen}
    \dot{\textbf{z}}_i(t) = F_i \textbf{z}_i(t) + \sum_{j=1}^N A_{ij} H  \textbf{z}_j(t),
\end{equation}
where the vector $\textbf{z}_i(t) \in \mathbb{R}^m$ represents the state of agent $i$ at time $t$.
The $m \times m$ matrices $F_i$ represent the individual dynamics of each agent $i = 1,\ldots,N$ when uncoupled.
The $m \times m$ matrix $H$ represents the coupling between agents.
The $N \times N$ adjacency matrix $A$ represents the graph $\mathcal{G}$.
A special case of such dynamics is the second order consensus described in the subsequent example.

\begin{definition}\label{def:identical}
  We say the agents have \emph{identical individual dynamics} if $F_i=F$, $i=1,...,N$.
  A weaker assumption is that agents have \emph{orbit-identical individual dynamics} if $F_i=F^k$, for all  $i \in {\mathcal V}_k$.
  In what follows we will proceed under the assumption of orbit-identical individual dynamics.
\end{definition}
The global system of equations are written by introducing the $Nm$-dimensional vector $\textbf{z}(t)=[\textbf{z}_1(t)^T, \textbf{z}_2(t)^T, ..., \textbf{z}_N(t)^T   ]^T$ and combining each node's contribution of Eq.\ \eqref{eq:gen},

\begin{equation} \label{eq:genn}
  \dot{\textbf{z}}(t) = \Bigl[ \sum_{k=1}^q J^k \otimes F^k  + A \otimes H \Bigr]  {\textbf{z}}(t),
\end{equation}
where the diagonal matrix $J^k$ has entries $J^k_{ii}=1$ if $\bar{i}=k$, $J^k_{ii}=0$ otherwise, and the symbol $\otimes$ indicates the Kronecker product of matrices.
    
\begin{remark}    
  For the case of identical individual dynamics, Eq. \eqref{eq:genn} becomes $\dot{\textbf{z}}(t)= \Bigl[ I_N \otimes F  + A \otimes H \Bigr]  {\textbf{z}}(t)$.
\end{remark}

\begin{definition}
  Consider orbit-identical individual dynamics as stated in Definition \ref{def:identical}.
  The set of states such that ${\textbf z}_i={\textbf z}_j$  for $\overline{i}=\overline{j}$ define an invariant manifold, which we call the \textbf{group consensus manifold.}
\end{definition}
  
Under the assumption of orbit identical individual dynamics, the dynamics on the group consensus manifold is governed by the \textit{quotient network} dynamics,
\begin{equation} \label{eq:q}
  \dot{\textbf{q}}_k (t)= F^k {\textbf{q}}_k(t) + \sum_{\ell = 1}^q Q_{k\ell} H {\textbf{q}}_{\ell}(t), \quad k = 1,\ldots,q, 
\end{equation}
where ${\textbf{q}}_k (t)$ is now the state of orbit $k=1,...,q$.
The $q \times q$ quotient matrix $Q$ is equal to 
\begin{equation}\label{eq:qnet}
  Q=(E^T E)^{-1} E^T A E = E^\dagger A E,
\end{equation}
where $E$ is the $N \times q$ indicator matrix, i.e., $E_{ij}$ is equal to one if node $i$ is in orbit $j$ and is equal to zero otherwise \cite{schaub2016graph}.

The quotient network describes the evolution of the system on the group consensus manifold.
In the case in which all the agents in the same cluster were given the same initial condition, the quotient network dynamics would provide the exact time evolution of all the network agents.

The quotient network dynamics is stable if the largest eigenvalue of the matrix
\begin{equation}
  \text{diag} \{F^1,\ldots,F^q \} + Q \otimes H
\end{equation}
is negative.
    
\begin{remark}
Under the assumption of orbit identical individual dynamics, the quotient network dynamics also describes the time evolution of an average state for all the nodes in the same orbit \cite{sorrentino2019symmetries},
\begin{equation}\label{eq:average}
    \textbf{q}_k (t)= \frac{1}{N_k} \sum_{i \in {\mathcal V}_k}  \textbf{z}_i (t)
\end{equation}
\end{remark}
\textcolor{sorren}{To see this, assume $F_i=F^k$, for all  $i \in {\mathcal V}_k$. Then, sum Eq.\ \eqref{eq:gen} over $i \in {\mathcal V}_k$ and divide by $N_k$. Recall that $\sum_{j \in {\mathcal V}_l} A_{ij}= Q_{kl}$ for each $i \in {\mathcal V}_k$ and $j \in {\mathcal V}_l$ (Remark 1). Then  Eq.\ \eqref{eq:q} follows, with the definition of $\textbf{q}_k (t)$ given in Eq.\ \eqref{eq:average}.}

As we will see, under appropriate conditions, the set of equations \eqref{eq:gen} and \eqref{eq:genn} will admit group consensus.
A set of agents may converge on group consensus on either a stable, or unstable, or marginally stable trajectory.
It is important to note that we are not concerned whether the entire system is asymptotically stable.
Instead, we focus our attention on the stability of each agent with respect to the group consensus state, where all the nodes in its orbits have reached consensus.

\begin{definition}
  We say that the nodes in orbit $\mathcal{V}_k$ have achieved group consensus if $\lim_{t \rightarrow \infty} \| {\textbf z}_i(t) - {\textbf z}_j(t) \|=0$ for all $i$ and $j$ in $\mathcal{V}_k$.
  As can be seen from this definition, group consensus is possible for either stable, marginally stable, or unstable node dynamics, as long as the trajectories converge to each other.
\end{definition}

\begin{remark}
Let us now assume identical individual dynamics. One may be tempted to diagonalize the matrix $A=V \Lambda V^{-1}$, where the matrix $\Lambda$ has diagonal entries $\lambda_1$, $\lambda_2$, ... $\lambda_N$. 
By pre-multiplying Eq.\ \eqref{eq:genn} by $V^{-1} \otimes I_m$ and introducing the vector $\textbf{e}=[\textbf{e}_1^T, \textbf{e}_2^T,..., \textbf{e}_N^T]^T=V^{-1} \otimes I_m \textbf{z} $, the transformed dynamics appears in $m$-dimensional blocks of the form,
\begin{equation} \label{eq:block}
   \dot{\textbf{e}}_i(t)= [F+ \lambda_i H ] {\textbf{e}}_i(t), \quad i = 1,\ldots,N
\end{equation}
While this exercise may provide some insight into the overall stability of the system, it may not allow one to predict the emergence of group consensus.
Assume for example that some of the blocks $[F+ \lambda_i H ]$ were found to be non-Hurwitz.
Then one could not conclude whether (i) group consensus is not achieved or (ii) group consensus is achieved but on a solution that is either diverging or a limit cycle.
We will see that there is a transformation of the system dynamics, provided by group theory, which is more appropriate in terms of characterizing group consensus.
\end{remark}

From knowledge of the group of symmetries of the  network, we can compute the irreducible representations (IRRs) of the symmetry group of the network.
This defines a transformation $T$ into the so called IRR coordinate system (see \cite{pecora2014cluster}).
{The transformation matrix $T$ is orthogonal. The first $q$ rows of the matrix $T$ are such that $T_{ki}=\sqrt{N_k}^{-1}$ if node $i$ is in cluster $k$ and $T_{ki}=0$ otherwise.
These rows describe motion that is parallel to the consensus manifold.
The remaining rows instead describe motion that is orthogonal to the consensus manifold and thus they describe its stability.}

{
Each one of the rows of the matrix $T$ is associated with a specific cluster, namely a total of $N_1$ rows are associated to cluster $1$, a total of $N_2$ rows are associated to cluster $2$, and so on.
If a row of the matrix $T$ is associated with cluster $k$, it means all the $i$ entries of that row are zero for $i$ not in cluster $\mathcal{V}_k$.}

\textcolor{sorren}{By defining the transformed state ${\tilde{\textbf{z}}}(t)=(T \otimes I_m)  \textbf{z}(t)$,  Eq.\ \eqref{eq:genn} can be rewritten,}
\begin{equation} \label{eq:genn1}
  \dot{\tilde{ \textbf{z}}}(t)= \Bigl[ \sum_{k=1}^q J^k \otimes F^k  + B \otimes H \Bigr]  {\tilde{\textbf{z}}}(t) = \hat{B} \tilde{\textbf{z}}(t),
\end{equation}
where  $T J^k T^T=J^k$, and $T A T^T$ is equal to the block-diagonal matrix $B$ \cite{pecora2014cluster}. \textcolor{sorren}{To see that $T J^k T^T=J^k$, recall that 
 the matrix $J^k$ is diagonal with ones on the diagonal corresponding to those entries in the $k$ cluster and zeros elsewhere.
  We examine the $(i,j)$ entry of the product $TJ^kT^T$,
  \begin{equation}
    \begin{aligned}
      (TJ^kT^T)_{ij} &= \sum_{\ell = 1}^n T_{i\ell} J_{\ell\ell}^k T_{j\ell}\\
      &= \sum_{\ell = 1}^n \delta_{\bar{\ell} \bar{i}} \delta_{\bar{\ell} \bar{j}} \delta_{\bar{\ell} k} T_{i \ell} J^k_{\ell \ell} T_{j \ell}
    \end{aligned}
  \end{equation}
  where $\bar{\ell}$ returns the cluster index of node $\ell$ and $\delta_{\bar{\ell} \bar{i}}$ returns 1 if node $\ell$ is in the same cluster as node $i$ and 0 otherwise.
  This implies that if either $i$ or $j$ is not in cluster $k$ the sum will be equal to zero.}
  
  \textcolor{sorren}{
  Now, consider just the set of indices $i$ and $j$ such that $\bar{i} = \bar{j} = k$.
  The summation becomes,
  \begin{equation}
    (TJ^kT^T)_{ij} = \sum_{\ell \in \mathcal{V}_k} T_{i \ell} T_{j \ell} J_{\ell \ell}^k = \sum_{\ell \in \mathcal{V}_k} T_{i \ell} T_{j \ell} = \delta_{ij}, \quad \bar{i} = \bar{j} = k
  \end{equation}
  where we used the property $J_{\ell\ell}^k = 1$ if $\bar{\ell} = k$ and that the rows of $T$ are orthonormal.
  It follows that $TJ^kT^T = J^k$.}

We can write the block-diagonal matrix $B$ as a direct sum $\oplus_{s=1}^S I_{d_s} \otimes C_s$ , where $C_s$ is a (generally complex) $p_s \times p_s$ matrix with $p_s$ the multiplicity of the $s$ IRR of the permutation group
representation, $S$ the number of IRRs and $d_s$
the dimension of the $s$ IRR, so that $\sum_{s=1}^S d_s p_s=N$ \cite{pecora2014cluster}. 
The trivial representation ($s=1$), which is associated with the motion in the synchronization manifold has $p_1=q$.
Each one of the remaining representations $s=2,...,S$ is associated with either: (i) an individual cluster or (ii) a set of \textit{intertwined clusters} \cite{pecora2014cluster}.
Accordingly, we say that in each irreducible representation is present either a cluster or a set of clusters.
The stability of each cluster depends on the maximum eigenvalue associated with each IRR in which the cluster is present.

This has important consequences in terms of the transformed linear dynamics $\tilde{z}(t)$.
The first consequence is that the transformed state vector is partitioned into two parts, $\tilde{\textbf{z}}(t) = [\tilde{\textbf{z}}_{para}^T(t), \quad \tilde{\textbf{z}}_{orth}^T(t)]^T$, where $\tilde{\textbf{z}}_{para}(t) \in \mathbb{R}^q$ describes the motion along the group consensus manifold and $\tilde{\textbf{z}}_{orth}(t) \in \mathbb{R}^{N-q}$ describes the motion orthogonal to the group consensus manifold.
The block diagonal matrix $\hat{B}$ decouples the motion along the group consensus manifold and the motion orthogonal to it.
\begin{equation}
  \left[ \begin{array}{c}
    \dot{\tilde{\textbf{z}}}_{para}(t) \\ \dot{\tilde{\textbf{z}}}_{orth}(t)
  \end{array} \right] = \left[ \begin{array}{cc}
    \hat{B}_{para} & O_{q \times N-q}\\
    O_{N-q \times q} & \hat{B}_{orth}
  \end{array} \right] \left[ \begin{array}{c}
    \tilde{\textbf{z}}_{para}(t) \\ \tilde{\textbf{z}}_{orth}(t)
  \end{array} \right]
\end{equation}
We define $\lambda_{para}^{\max}$ to be the maximum real part of the eigenvalues of $\hat{B}_{para}$ and $\lambda_{orth}^{\max}$ to be the maximum real part of the eigenvalues of $\hat{B}_{orth}$.
Stability of the motion along the group consensus manifold is determined by the sign of $\lambda_{para}^{\max}$. 
If the system corresponding to orthogonal motion, $\hat{B}_{orth}$, is Hurwitz, that is $\lambda_{orth}^{\max} < 0$, then the group consensus manifold is stable and any perturbation orthogonal to the group consensus manifold will decay to zero, independent of the behavior along the manifold.
This implies that it is possible for the original system to be marginally stable, or unstable, yet still achieve group consensus.

\begin{lemma}
The two matrices $Q$ and $\hat{B}_{para}$ are similar.
\end{lemma}

\begin{proof}
 The block $\hat{B}_{para}= \tilde{T} A \tilde{T}^T$, where the $q \times N$ matrix $\tilde{T}$ is composed of the first $q$ rows of the matrix $T$. Moreover, $\tilde{T}=(E^T E)^{-\frac{1}{2}} E^T$. Hence $\hat{B}_{para}=(E^T E)^{\frac{1}{2}} Q (E^T E)^{-\frac{1}{2}}$.
\end{proof}

Another consequence is that based on the block-diagonal structure of the matrix $\hat{B}_{orth}$, the vector $\tilde{\textbf{z}}_{orth}(t)$ may be partitioned into a number of vectors evolving independently of each other (each one corresponding to a non-trivial irreducible representation of the graph automorphism group.)
This implies that for a given graph, certain clusters may achieve \textit{isolated group consensus}, while others may not, as will become apparent from the example that follows. 
This is a significant difference with respect to the current literature on group consensus \cite{yu2012group,qin2016leaderless} where group consensus is deemed to either occur or not occur.

As discussed previously, the block diagonalized system $\hat{B}$ consists of a set of independently evolving systems, each of which can be assigned to a set of clusters.
Let $\hat{B}_k$ be the $k$th block corresponding to the system $\dot{\tilde{\textbf{z}}}_k = \hat{B}_k \tilde{\textbf{z}}_k(t)$.
Also, let $\mathcal{I}_k$ be the set of clusters associated with this block.
Note that $k = 1$ corresponds to the parallel, group consensus, motion so $\hat{B}_1 = \hat{B}_{para}$ and $|\mathcal{I}_k| = q$.
As for the additional, orthogonal blocks we note that each cluster, $\mathcal{V}_{\ell}$, appears in $N_{\ell}-1$ orthogonal blocks.
If one is interested in only ensuring that a specific set of clusters achieve isolated cluster consensus, it is enough to ensure that the set of orthogonal blocks in which the clusters appear are Hurwitz.
This is explored in some detail in the following example.

\begin{figure}
    \centering
    \includegraphics{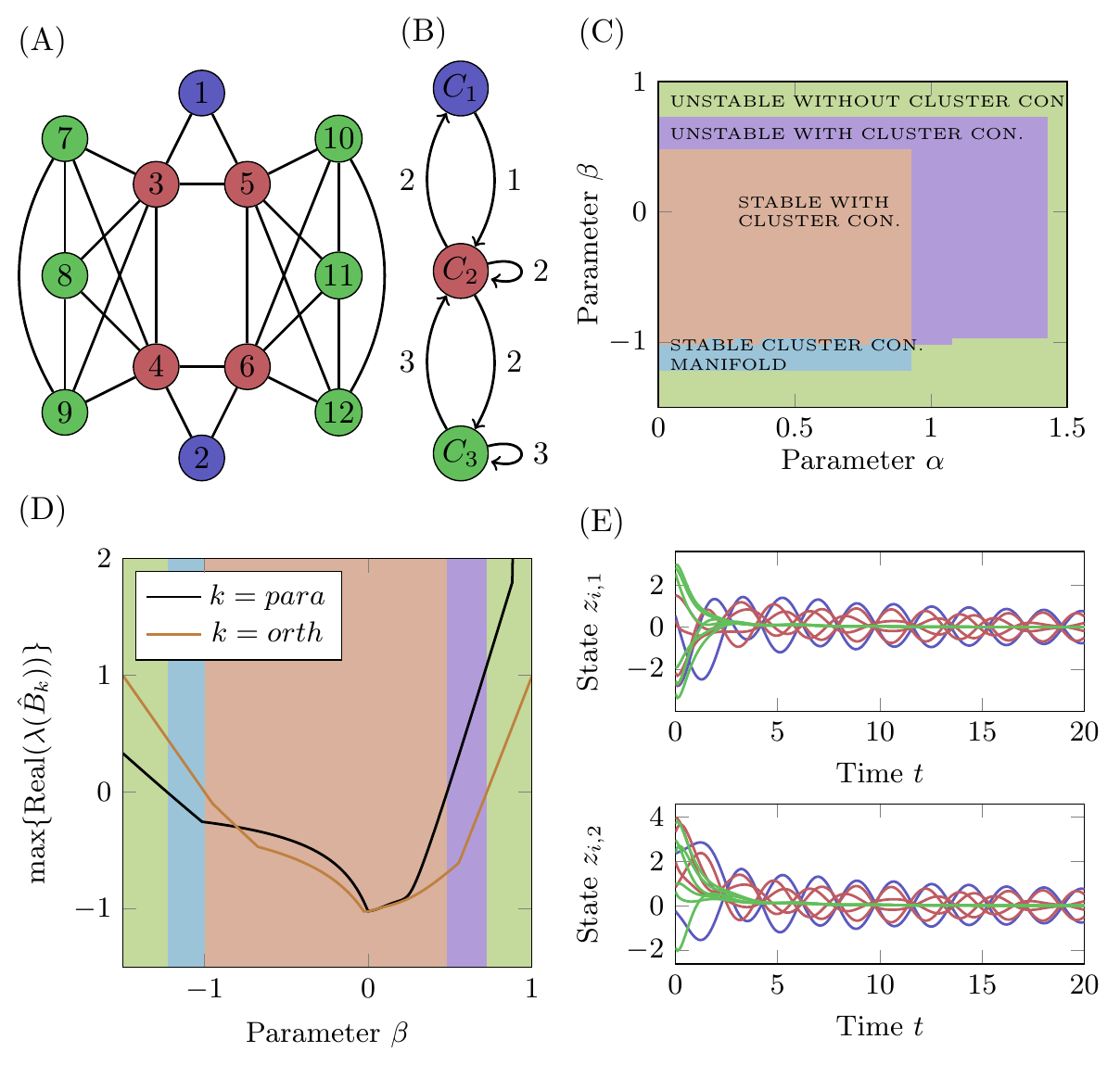}
    \caption{
      Using the $\hat{B}$ matrix to determine group consensus.
      (A) The diagram of the graph of interest.
      This graph has $N = 12$ nodes and $q = 3$ orbits.
      The nodes are colored according to their orbits.
      (B) The quotient network of the graph in (A) partitioned according to its orbits.
      (C) The stability regions in parameter space are colored according to the maximum real part of the eigenvalues of $\hat{B}_{para}$ and $\hat{B}_{orth}$, respectively.
      The green region represents when both the quotient dynamics Eq.\ \eqref{eq:q} is unstable, and group consensus will not occur,  the purple region represents when Eq.\ \eqref{eq:q} is unstable, but group consensus will occur, and the red region represents when Eq.\ \eqref{eq:q} is stable and group consensus will occur.
      The blue region represents a condition in which the quotient dynamics is stable, but group consensus is not. \textcolor{sorren}{In this case, group consensus only occurs if the initial conditions are on the group consensus manifold, then the system  decays to zero.}
      (D) We set $\alpha = 0.2$ and vary $\beta$ from $-1.5$ to $1$. The maximum real part of the eigenvalues of $\hat{B}_{para}$ and $\hat{B}_{orth}$ are plotted, and the stability regions are colored according to the diagram in (C).
      (E) We set $\alpha = 0.2$ and $\beta = -1$ which is on the boundary of the blue and red regions.
      For these values, $\hat{B}_{para}$ is negative definite, but $\hat{B}_{orth}$ is marginally stable.
      The marginally stable block in $\hat{B}_{orth}$ corresponds to the red cluster.}
    \label{fig:diagram}
\end{figure}

Consider the following example with $N = 12$ nodes, $m = 4$ states per nodes, and $q = 3$ orbits in the graph's automorphism group.
A diagram of the graph is shown in Fig. \ref{fig:diagram}(A) and the associated quotient graph is shown in Fig. \ref{fig:diagram}(B).
The graph is drawn such that the symmetries are visually clear, but note that usually the orbits must be found using computational group theory packages, such as Sage \cite{sagemath}.
The nodal dynamics matrices for the nodes in each orbit are,
\begin{equation*}
  \begin{aligned}
    F^1 = \left[ \begin{array}{cc}
      O_2 & I_2 \\ -2 I_2 & -2 I_2
    \end{array} \right], && F^2 = \left[ \begin{array}{cc}
      O_2 & I_2 \\ -4 I_2 & -4 I_2
    \end{array} \right], && F^3 = \left[ \begin{array}{cc}
      O_2 & I_2 \\ -6 I_2 & -6 I_2
    \end{array} \right]
  \end{aligned}
\end{equation*}
The coupling matrix $H$ that describes the connections between nodes is,
\begin{equation*}
  H = \left[ \begin{array}{cc}
    O_2 & O_2 \\ \alpha I_2 & \beta \textbf{1}_2 \textbf{1}_2^T,
  \end{array} \right]
\end{equation*}
where $\alpha$ and $\beta$ are parameters to be chosen to affect the stability of the system.
The adjacency matrix can be determined from the diagram of the graph in Fig. \ref{fig:diagram}(A) and is shown in Fig. \ref{fig:A} where the rows are colored according to the orbit in which each node appears.
\begin{figure}[h!]
\centering
\includegraphics[scale=1]{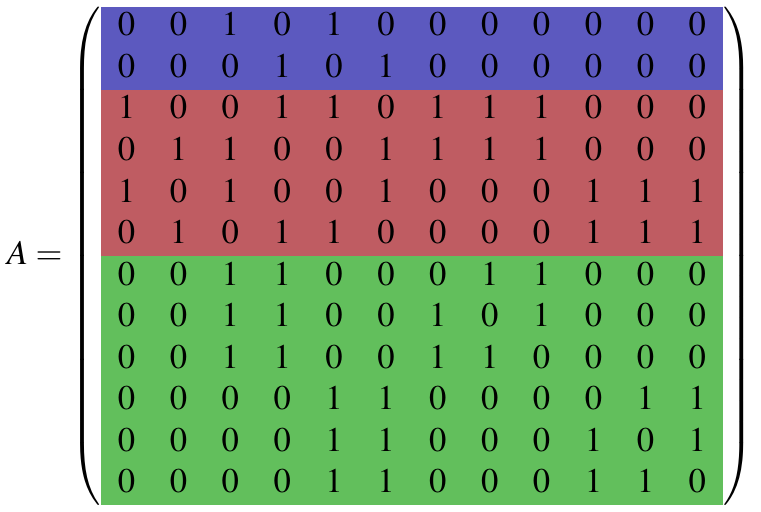}
\caption{The matrix $A$ corresponding to the network in Fig.\ \ref{fig:diagram}. Each row $i$ of the matrix $A$ is colored according to the orbit to which node $i$ and consistently with the coloring in Fig.\ \ref{fig:diagram}.}
\label{fig:A}
\end{figure}

The transformation matrix $T$ that block diagonalizes the adjacency matrix was found using the method discussed and code provided in \cite{pecora2014cluster} and is shown in Fig. \ref{fig:T}.
The rows are colored according to the cluster in which the non-zero entries appear, i.e., the first row is colored green as the non-zero entries appear in columns 7-12, the same indices as the green nodes in Fig. \ref{fig:diagram}(A).
%

\begin{figure}[h!]
\centering
\includegraphics[width=\textwidth]{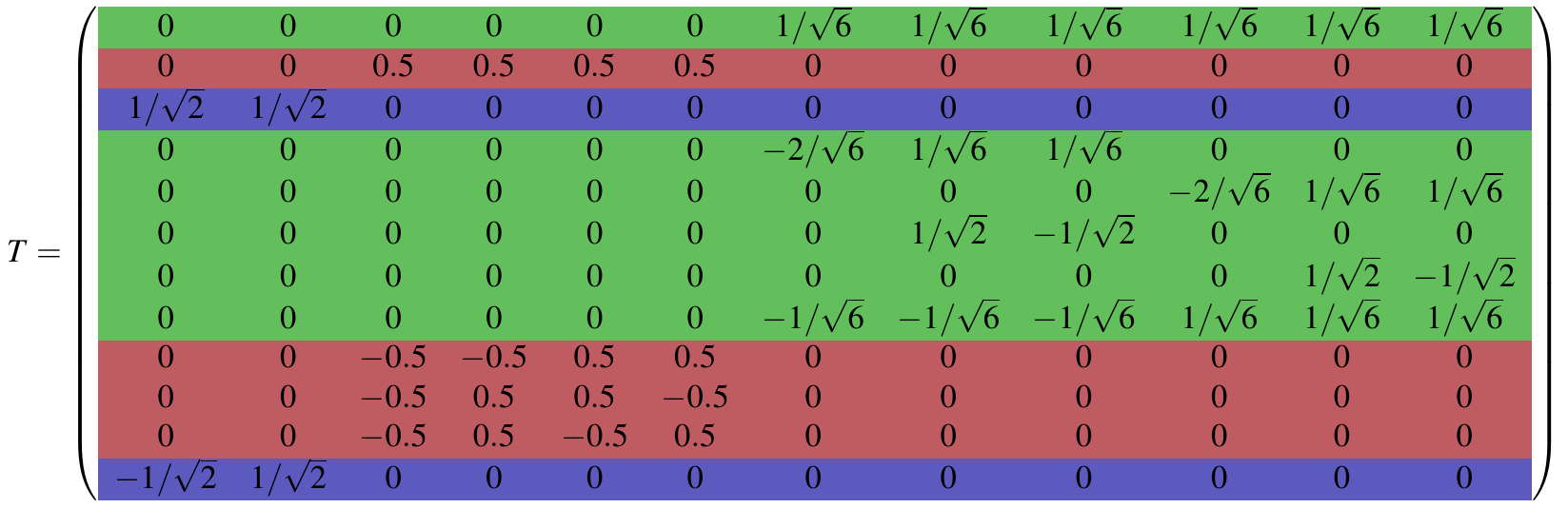}
\caption{The matrix $T$ corresponding to the network in Fig.\ \ref{fig:diagram}. Each row $i$ of the matrix $T$ is colored according to the orbit to which node $i$ belongs and consistently with the coloring in Fig.\ \ref{fig:diagram}.}
\label{fig:T}
\end{figure}


The block diagonalized matrix $B = TAT^T$ is shown in Fig. \ref{fig:B} with the blocks colored according to which clusters' motion they represent, i.e., the same color pattern that appears in the tranformation $T$ shown in Fig. \ref{fig:T}.
\begin{figure}[h!]
  \centering
  \includegraphics[scale=1]{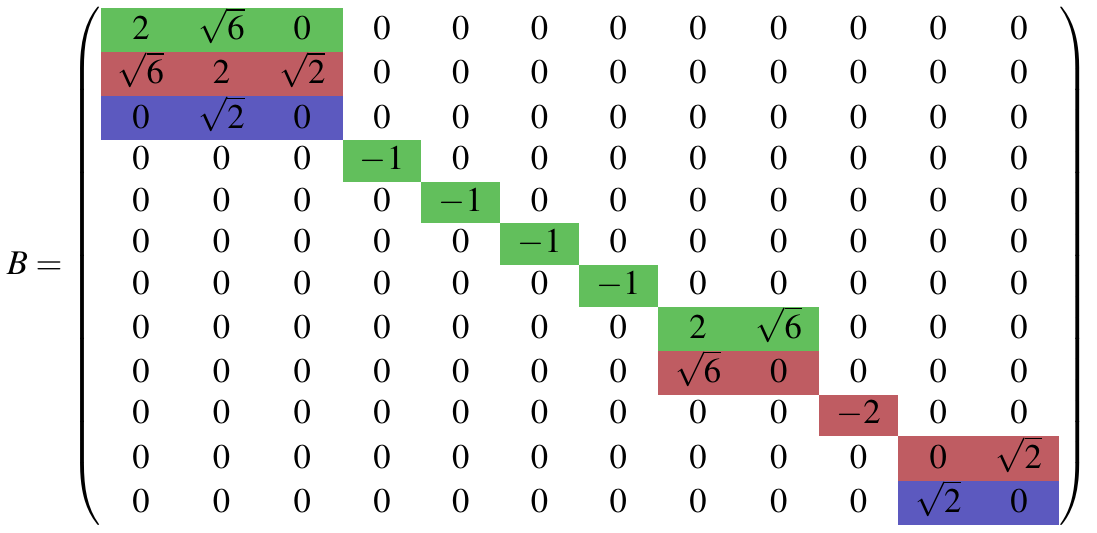}
  \caption{The matrix $B=T A T^T$ corresponding to the network in Fig.\ \ref{fig:diagram}. Each row $i$ of the matrix $T$ is colored according to the rows of the matrix $T$ to which that block corresponds.}
  \label{fig:B}
\end{figure}

As can be seen, $B_{para}$ is a $3 \times 3$ block, corresponding to the trivial irreducible representation ($s=1$);  $B_{orth}$ is the direct sum of $7$ blocks.
From top to bottom, there are four IRRs: the IRR $s=2$ has dimension $d_2=1$ and multiplicity $p_2=4$, the IRR $s=3$ has dimension $d_3=2$ and multiplicity $p_3=1$, the IRR $s=4$ has dimension $d_4=1$ and multiplicity $p_4=1$, and the IRR $s=5$ has dimension $d_5=2$ and multiplicity $p_5=1$.
The green cluster is present in the IRRs 2 and 3, the red cluster in IRRs 3, 4, and 5 and the blue cluster in IRR 5.
Looking at Fig.\ \ref{fig:diagram}, IRR 3 corresponds to a left-right symmetry breaking that affects the red and green clusters but not the blue one, while  IRR 5 corresponds to a top-bottom symmetry breaking that affects the red and blue clusters but not the green one.

Varying the parameters $\alpha$ and $\beta$, we shade the regions of stability in Fig, \ref{fig:diagram}(C) according to the maximum real part of the eigenvalues of $\hat{B}_{para}$ and $\hat{B}_{orth}$.
The green region represents values of $\alpha$ and $\beta$ where both $\lambda_{para}^{\max} > 0$ and $\lambda_{orth}^{\max} > 0$ which means both the system is unstable and group consensus will not occur.
The purple region represents values of $\alpha$ and $\beta$ where $\lambda_{para}^{\max} > 0$ so that the group consensus manifold is unstable, but $\lambda_{orth}^{\max} < 0$ so group consensus will occur.
The red region represents values of $\alpha$ and $\beta$ where both $\lambda_{para}^{\max} < 0$ and $\lambda_{orth}^{\max} < 0$ so that the system is stable and group consensus will occur.
Finally, the blue region represents values of $\alpha$ and $\beta$ where $\lambda_{para}^{\max} < 0$ so the group consensus manifold is stable, but $\lambda_{orth}^{\max} > 0$ so group consensus will not occur if the state is perturbed away from the group consensus manifold.
To examine the behavior of $\lambda_{para}^{\max}$ and $\lambda_{orth}^{\max}$, we set $\alpha = 0.2$ in Fig. \ref{fig:diagram}(D) and vary $\beta$ from $-1.5$ to $1$.
We color the background of the plot according to the regions in Fig. \ref{fig:diagram}(C).
To examine the temporal behavior of the system, in Fig. \ref{fig:diagram}(E) we plot the time traces of $z_{i,1}(t)$ and $z_{i,2}(t)$ for $\alpha = 0.2$ and $\beta = -1$ for each of the nodes, colored according to the nodes' clusters.
This point is at the boundary of the red and blue regions in Fig. \ref{fig:diagram}(D) which corresponds to $\lambda_{para}^{\max} < 0$ and $\lambda_{orth}^{\max} = 0$, i.e., the group consensus manifold is stable but orthogonal perturbations will not damp out.
As $\hat{B}$ is block diagonal, we can find the particular block that is marginally stable, which for this example is the $1 \times 1$ block $-2$ in the tenth row of $B$.
Tracing back to the transformation matrix $T$, we see that the corresponding row represents the perturbation that splits nodes $3$ and $6$ from nodes $4$ and $5$ in cluster $\mathcal{V}_2$.
Nonetheless cluster $\mathcal{V}_3$ still converges, i.e., this is an example of isolated group consensus.

\begin{remark}
  Consider the case of identical individual dynamics.
  Then all the blocks (either orthogonal or parallel) corresponding to an IRR with dimension $d>1$, can be further diagonalized into $d$ $m$-dimensional systems of the form $\dot{\hat{\textbf z}}_i(t)= [F+ \lambda_i H ] {\hat{\textbf{z}}}_i(t)$.
  Hence, the system dynamics \eqref{eq:gen} is transformed again into $N$ equations of the form of Eqs.\  \eqref{eq:block}, in the eigenvalues $\lambda_i$ of the matrix $A$, $i=1,..,N$.
  However, there is  an important advantage of the group theoretical approach.
  Namely, the matrix $T$ carries information on which eigenvalues are associated with motion parallel to the consensus manifold and which eigenvalues are associated with motion transverse to the manifold and for the latter ones, which ones are associated with either stability of a given cluster or stability of a given set of intertwined clusters. With this knowledge,  the set of equations \eqref{eq:block} will provide detailed information on stability of the quotient network dynamics and on stability of group consensus for given clusters of interest.
\end{remark}

\section{Conclusion}

In this paper, we have studied the group consensus problem from the perspective of graph automorphisms.
We have shown how the block diagonalizing transformation can decouple the motion along the consensus manifold from the motion orthogonal to the consensus manifold.
More importantly, the transformation splits the orthogonal motion into disjoint systems of equations corresponding to the types of orthogonal motion, which allows one to choose coupling protocols to selectively allow some clusters to reach consensus, which we call isolated group consensus.
We also show we can achieve group consensus in the absence of intra-cluster edges, different from the current methodologies for the group consensus problem \cite{yu2012group,qin2016leaderless}.

In the example presented, we show that for our particular choice of coupling protocol with two tunable control parameters, we are able to achieve group consensus whether or not the system is stable.
We also choose values for the two control parameters which allow for individual group consensus.
For other coupling protocols, the steps taken through the example can be used to find the range of tunable control parameters which allow for group consensus.
Alternatively, the block diagonalizing transformation can be used to construct Lyapunov functions for individual orthogonal components which can be tailored to the particular form of each block for other choices of coupling protocols, with potentially more tunable control parameters.

For a given network topology, this work provides the series of steps to be taken to analyze the conditions under which either complete group consensus or individual group consensus can be achieved, regardless of the stability of the entire system.

\begin{acknowledgments}
This work was supported by the National Science Foundation through grant (Grant No. 1727948) and the Office of Naval Research through ONR Award No. N00014-16-1-2637.
\end{acknowledgments}

\end{document}